\documentclass[12pt, a4paper,reqno]{amsart}

\usepackage[margin=1in]{geometry}
\usepackage{amssymb}
\usepackage{enumerate}
\usepackage{thmtools}
\usepackage{hyperref}
\usepackage{doi,url}
\usepackage{dsfont}
\usepackage{mathabx}
\usepackage[normalem]{ulem}
\usepackage{cancel}
\usepackage{bm}

\usepackage[dvipsnames]{xcolor}
\usepackage{tikz}
\usepackage{pgfplots}
\usepgfplotslibrary{fillbetween}
\pgfplotsset{compat=1.16,width=10cm}
 \usepackage{caption}
\usepackage{float}
\usetikzlibrary{patterns}

\definecolor{purple}{rgb}{.5,0,1}
\definecolor{orange}{rgb}{1,.5,0}
\definecolor{pink}{rgb}{1,0,.5}

\usepackage{fancyvrb}
\usepackage{amsmath,amsthm}
\usepackage{enumitem}

\usepackage[
  style=numeric,
  sorting=nyt,
  maxbibnames=99,
  giveninits=true,
  url=true,
  doi=true,
  backend=biber
]{biblatex}

\renewbibmacro{in:}{}
\AtEveryBibitem{\clearfield{number}}
\DeclareFieldFormat*{volume}{\mkbibbold{#1}}
\DeclareFieldFormat[article]{pages}{#1}

\numberwithin{equation}{section}

\newtheorem{theorem}{Theorem}[section]
\newtheorem{proposition}[theorem]{Proposition}
\newtheorem{lemma}[theorem]{Lemma}

\newtheorem{remark}[theorem]{Remark}

\newtheorem*{theorem*}{Theorem}
\newtheorem*{lemma*}{Lemma}
\newtheorem*{remark*}{Remark}

\DeclareMathOperator{\supp}{supp}

\DeclareMathOperator{\dist}{dist}

\renewcommand\H{\mathcal{H}}

\newcommand\R{\mathbb R}
\newcommand\N{\mathbb N}
\newcommand\C{\mathbb C}
\newcommand\Z{\mathbb Z}

\newcommand\cW{\mathcal{W}}

\newcommand\e{\mathrm{e}}

\renewcommand\P{\mathbb P}
\newcommand\E{\mathbb E}
\newcommand\cE{\mathcal{E}}

\newcommand\cN{\mathcal{N}}

\newcommand\cX{\mathcal{X}}

\newcommand\cH{\mathcal{H}}

\newcommand\vphi{\varphi}

\renewcommand{\d}{\mathrm{d}}

\newcommand{\pr}{\prime}

\newcommand\what{\widehat}

\newcommand{\bom}{{\boldsymbol{{\omega}}}}

\newcommand\beq{\begin{equation}}
\newcommand\eeq{\end{equation}}
\newcommand\be{\begin{equation}\begin{aligned}}
\newcommand\ee{\end{aligned}\end{equation}}

\newcommand{\abs}[1]{\left\lvert #1 \right\rvert}
\newcommand{\norm}[1]{\left\lVert #1 \right\rVert}

\newcommand{\set}[1]{\left\{ #1 \right\}}
\newcommand{\pa}[1]{\left( #1 \right)}

\newcommand{\fl}[1]{\left\lfloor #1 \right\rfloor}

\newcommand{\cl}[1]{\lceil #1 \rceil}

\newcommand\La{\Lambda}

\newcommand{\eq}[1]{\eqref{#1}}

\newcommand{\up}[1]{^{\left(#1\right)}}

\newcommand{\qtx}[1]{\quad\text{#1}\quad}
\newcommand{\mqtx}[1]{\; \ \text{#1}\; \  }
\newcommand{\sqtx}[1]{\;\text{#1}\;}

\newcommand{\bD}{\bm{\Delta}}

\newcommand{\tfd}{\pa{1- \tfrac{1}{\Delta}}}
\newcommand{\fd}{1- \frac{1}{\Delta}}
\newcommand{\sfd}{1- \tfrac{1}{\Delta}}
\newcommand{\nfd}{\pa{1- \frac{1}{\Delta}}}

\newcommand{\bpartial}{\bm{ \partial}}

\newcommand{\bmo}{\bm{\omega}}

\newcommand{\prr}{{\pr\pr}}

\newcommand{\ident}{\mathds{1}}

\begin{document}

\title[MBL  for the random XXZ chain in  fixed intervals]{Many-body localization  for the random XXZ spin chain in fixed energy intervals}
\author{Alexander Elgart}
\address[A. Elgart]{Department of Mathematics, Virginia Tech; Blacksburg, VA, 24061, USA}
 \email{aelgart@vt.edu}

\author{Abel Klein}
\address[A. Klein]{Department of Mathematics, University of California, Irvine;   
Irvine, CA 92697,  USA}
 \email{aklein@uci.edu}

\thanks{A.E. was  supported in part by the National Science Foundation under grant DMS-2307093.}

\date{Version of \today}

\begin{abstract}
 It is shown that  the infinite random Heisenberg XXZ spin-$\frac12$ chain exhibits slow propagation of information  (logarithmic light cone), a   key signature of many-body localization (MBL),   in any fixed energy interval at the bottom of the spectrum.
 The relevant parameter regime, which covers both weak interaction and strong disorder, is determined solely by the  energy interval.
 \end{abstract}

\keywords{Many-body localization, MBL, random XXZ spin chain, slow propagation of information}

\subjclass[2000]{82B44, 82C44, 81Q10, 47B80, 60H25}

\maketitle

\tableofcontents

\section{Introduction}

 In this paper we   prove that the infinite random Heisenberg XXZ spin-$\frac12$ chain exhibits slow propagation of information, in the form of a logarithmic light cone, within any fixed energy interval at the bottom of the spectrum. Slow propagation of information  is considered a primary signature of many-body localization (MBL). The applicable parameter regime encompasses both strong disorder and weak interaction, and  depends only the fixed energy interval.

A finite spin-$\frac 12$ chain is given by a self-adjoint operator $H_L$ in a finite dimensional tensor product Hilbert space $\mathcal H_L=\otimes_{i=-L}^L \C_i^2$, where $\C_i^2$ indicates the $i$-th copy of $\C^2$ (the $i$-th spin-$\frac 12$)  and the parameter $L\in\N$ describes the chain's length. Physically, it represents a one-dimensional array of spins at sites $-L,\ldots,L$. The infinite analog of $H_L$ is $H_\Z$, acting on the (properly defined) Hilbert space $\mathcal H_\Z \simeq  \otimes_{i\in\Z} \C^2_i$. Such systems appear frequently in statistical and condensed matter physics and can be traced back to Heisenberg famous 1928 work, where he  introduced a spin model (named after him) to explain ferromagnetism. More recently, spin chains have become prototypical models in quantum information theory and quantum computing, among many other areas. The behavior of such systems (including a random variant of the Heisenberg XXZ spin chain we consider here) are generally poorly understood, even numerically, unless they are exactly solvable. One of the main difficulties lies in the exponentially fast growth of the dimension of $\mathcal H_L$ in $L$.

Starting from the mid-2000s, random spin chains emerged as prototypical models to determine whether many-body systems can exhibit an exotic phase of matter known as many-body localization (MBL). For single particle systems, it is known that randomness can cause a drastic change in material's dynamical properties -- it can lead to non-propagation  of the wave packets that are initially spatially localized. This effect was posited by Anderson in 1958 and is known as Anderson localization. It has been extensively studied mathematically since the 1980s and its mechanism is relatively well understood by now. 
However, it is not known whether the Anderson localization is a stable phase of the matter, in particular whether it survives the effects of particle-particle interaction, which is inherent in the realistic description of physical samples. 

The single-particle picture provides a satisfactory description of a many-body system's behavior only if the quantum data related to a single-particle observable does not show significant propagation to other particle sectors in the course of its dynamical evolution. The phenomenon of  non-spreading, or at least remarkably slow spreading of information, is considered one of the defining signatures of MBL. Disordered quantum spin chains are widely believed to be the most likely physical setting in which this type of restricted propagation can occur. A robust debate about the true existence and stability of the MBL phase has been actively taking place in the field of condensed matter physics for the last two decades, with substantial arguments presented both in favor of and against it.

MBL is usually associated with the localization of states with
a finite energy density.. In particular, when the system’s size grows to infinity (the thermodynamic limit), one is required to consider energy intervals that also grow to infinity. In this paper, we prove the weaker statement that MBL, manifested in the form of  slow propagation of information, occurs within fixed energy intervals at the bottom of the spectrum of the random ferromagnetic Heisenberg XXZ spin chain.

Localization of the ground state and  of a small number of excited states is commonly known in physics as   zero-temperature localization.  Our results on  localization in a fixed energy interval at the bottom of the spectrum therefore occupy  intermediate position between this zero-temperature localization regime and the full MBL regime. In the subsequent section, we introduce the necessary notation, formulate our result, and  provide a comparison with existing findings.

\section{The random XXZ spin chain and the  main result}
\subsection{Setting}
A single spin-$\frac 12$ chain is  a two-dimensional complex state space $\C^2$ with the standard basis consisting of the spin-up  $\uparrow\rangle = {\begin{bmatrix} \scriptstyle 1  \\ \scriptstyle 0  \end{bmatrix}} $ and spin-down $\downarrow\rangle = {\begin{bmatrix} \scriptstyle 0  \\ \scriptstyle 1 \end{bmatrix}} $ states.  We will denote by $\sigma^x, \sigma^y, \sigma^z$  the three Pauli matrices and by $\sigma^\pm=\frac 12( \sigma^x \pm i \sigma^y)$ the ladder operators.    We  interpret $\downarrow\rangle$ as a particle, with  $\mathcal{N} = \tfrac{1}{2} (\ident-\sigma^z)$,  the projection onto the spin-down state, called the  (local) number operator on $\C^2$.

Spin chains are arrays of  spins indexed by subsets $\La\subset\Z$. If $\La$ is finite, the corresponding state space is the tensor product Hilbert space $\mathcal H_\Lambda=\otimes_{i\in \Lambda} \cH_i$, where each $\cH_i$ is a copy of $\C^2$.  Given a vector $v\in \C^2$, we denote by $v_i$ its copy in $\cH_i$. 
 For infinite $\La$, we let  $\cH_{\La,0}$ be the vector subspace of $\bigotimes_{i\in \La}\cH_i$ spanned by    tensor products of the form  $\bigotimes_{i\in \La}\vphi_i$,  $\vphi_i\in \set{\uparrow\rangle_i,\downarrow\rangle_i}$, with a finite number of spin-downs, equipped with the  standard inner product on a tensor product space,  and let $\cH_\La$ be its Hilbert space completion.  A single spin operator $\sigma^\sharp$ acting on the $i$-th spin is lifted to  $\mathcal H_\Lambda$  by identifying it with  $\sigma^\sharp \otimes \mathds{1}_{\La\setminus \set{i}}$, where $ \mathds{1}_{\La\setminus \set{i}}$ denotes the identity operator on $\cH_{\La\setminus \set{i}}$, i.e.,   it acts non-trivially only in the tensor product's $i$-th component. To stress the $i$-th dependence, we will denote this single spin operator by $\sigma^\sharp_i$. More generally,  if  $S\subset \La$, and $A_S$ is an operator on  $\cH_S$, we often  identify it with its natural embedding on $\cH_\La$, namely with  $A_S \otimes \mathds{1}_{\La\setminus S}$.

The random XXZ spin chain  on $\Z$   is given by  the Hamiltonian
  $ H=H_{\bmo} = H_0 +  \lambda V_{\bmo}$, where:
\begin{enumerate}
\item $H_0$ is the  disorder-free  Hamiltonian, given by 
\be
H_0 =  \sum_{i\in\Z} - \mathcal{N}_{i}\mathcal{N}_{i+1} -\tfrac{1}{2\Delta}\pa{\sigma_i^+\sigma_{i+1}^-+\sigma_i^-\sigma_{i+1}^+} +  \cN_i.
\ee
Here  $\Delta>1$ is the anisotropy   parameter, specifying the Ising phase ($\Delta=1$ selects the Heisenberg chain and  $\Delta=\infty$ corresponds to the  Ising chain).
\item    $\lambda >0$ is the disorder parameter and $V_{\bmo}$ is the random field  given by
\be
V_{\bmo}=\sum_{i\in\Z} \omega_i \mathcal{N}_i,
\ee
where
 $\bmo = \set{\omega_i}_{i\in\Z}$  is a family of independent identically distributed random variables whose  common probability  distribution $\mu$  is  absolutely continuous with a bounded density  and satisfies
\beq\label{mu}
  \set{0,1}\subset \supp \mu\subset[0,1] .
  \eeq 
\end{enumerate}
$H_{\bmo}$ is a self-adjoint operator on $\cH_\Z$ for a fixed $\bmo \in [0,1]^\Z$;  it is essentially self-adjoint on $\cH_{\Z,0}$.  Its spectrum satisfies
\be
\sigma(H_{\bmo}) =\set{0}\cup \pa{[1-\tfrac1\Delta,\infty)\cap \sigma(H_{\bmo})}.
\ee
The ground state energy $0$ is a simple eigenvalue of $H_{\bmo}$ with the corresponding eigenvector $\phi_\emptyset= \bigotimes_{i\in \Z}\uparrow\rangle_i$.

We will generally omit $\bom$ from the notation for the random operator, e.g. we simply write $H$ for $H_{\bmo}$, and always assume $\bmo \in [0,1]^\Z$.    It follows from the usual ergodicity arguments that
 \be
 \sigma(H)=\set{0}\cup[1-\tfrac1\Delta,\infty) \quad\text{with probability one}. 
 \ee

A {\it local observable} is a bounded operator $T$ on $\cH_\Z $  supported by a finite subset $\cX$ of $\Z$, that is,
 $T=T_{\cX}
\otimes \mathds{1}_{\mathcal H_{\Z\setminus \cX}}$, where $T_{\cX}$ is a bounded operator on $\cH_\cX$.  As already mentioned, we will usually identify $T$ with $T_{\cX}$. We call $\cX$ the {\it support} for $T$; note that the support  of $T$ is not unique.

Slow propagation of information is in reference to the {\it Heisenberg time evolution} of observables: Given the initial observable $T$, its evolution  is given by  $\tau_t(T)=\e^{itH}T\e^{-itH}$   at time $t$. In this paper, we restrict the time evolution to the energy window $[0,E]$, that is, we are interested in the behavior of the  random operator $S(T,t,E)= P_{E}\pa{\e^{itH}T\e^{-itH}}P_{E}$,  where     $P_{E}=\chi_{ (-\infty,E]}(H)=\chi_{ [0,E]}(H) $ is the Fermi projection.  The slow propagation of information is then expressed as the fact that for any local observable $T$ we can approximate $S(T,t,E)$ by $P_{E}T_t P_{E}$ for some local observable $T_t$, where the support of $T_t$ is not significantly larger than the support of $T$. The size of the support depends on the precision of the approximation and (weakly) on $t$. To quantify the growth of  the support  it is convenient to introduce the following notation:   Given  $ S\subset \Z$ and  $p\in\N^0$,   we set $[S]_p = \set{x\in\Z: \dist\pa{x,S}\le p}$.  Denoting by $\E$ the expectation with respect to the random variables, and setting $\langle t\rangle=1 + \abs{t}$, our main result is

\begin{theorem}[ Slow propagation of information-infinite volume]\label{thm:main}
 Fix parameters $\Delta_0>1$ and  $ \lambda_0 >0$.   Then for all   $E\ge 0$  there exist strictly positive constants $C_E, D_E,  \kappa_E, m_E$ (depending on $E$, $\Delta_0$, $ \lambda_0$) such that for all $\Delta \ge \Delta_0$ and $\lambda \ge \lambda_0$ with  $\lambda \Delta^2\ge D_E$ the following assertion holds: Given an observable $T$  supported on a finite  interval $\mathcal X\subset \Z$ with $\norm{T}\le 1$ and a scale $\ell\in \N$, for all $t\in \R$ there exists a  random observable ${T}_t={T}(\bmo,t,E,\ell)$, supported on  $ [\cX]_{\ell}$, such that
\be\label{eq:main}
\E\norm{ P_{E} \pa{\e^{itH}T\e^{-itH}-T_t} P_{E}}\le    C_{E} \langle t\rangle^{\kappa_E} \e^{- m_{E}\, \ell}.
\ee
 \end{theorem}

 \subsection{Logarithmic vs linear light cones}\label{subsec:cones}
 A consequence of the Lieb-Robinson bound \cite{LR,NachtSims} for a local spin Hamiltonian $H$  is that for any local observable $T$ supported on a finite  interval $\mathcal X\subset \Z$  with $\norm{T}\le 1$, given $\ell \in \N$ there exists an  observable $\mathcal T_{t}$, supported on $[\mathcal X]_\ell$, that satisfies 
\be\label{LR}\norm{\e^{itH}T\e^{-itH}-T_{t}}\le C\e^{-m(\ell-v|t|)},
\ee
where  $m>0$,  and $v>0$ is the velocity in the Lieb-Robinson bound.  For  systems defined on finite sets $\La\subset \Z$, this was originally  shown in \cite{BHV}, where the approximating observable $T_t$ was chosen to be the partial trace of $\e^{itH}T\e^{-itH}$ with respect to degrees of freedom  associated with the set $\La\setminus [\mathcal X]_\ell$. Consequently, the approximation in \cite{BHV} depends on the size of the system  since one averages over $\La\setminus [\mathcal X]_\ell$. 
A different proof is given in \cite[Theorem 3.4]{Young}, where $T_t=\e^{itH^{ [\mathcal X]_\ell}}T\e^{-itH^{ [\mathcal X]_\ell}}$ is volume independent, $H^{ [\mathcal X]_\ell}$ being an appropriately defined restriction of $H$ to the finite interval $ [\mathcal X]_\ell$. Although the result is stated and proven for systems defined on finite sets $\La\subset \Z$,
the proof is valid also for $\La=\Z$.

In physics, this result is interpreted as propagation of information   in the    {\it linear light cone}, that is,   at most linear in time propagation of information, up to  exponentially small corrections.  For  translation invariant systems, it is expected that this bound cannot be improved, i.e.,  information can indeed spread within the light cone. In contrast, for the random system, Theorem \ref{thm:main} establishes a much slower rate of spread of information. Specifically, information can only potentially escape the corresponding {\it logarithmic} light cone after a time $t\sim e^{c\ell}$ for a given value of $\ell$, which is significantly slower than the linear $t\sim \ell$ rate. Although the time dependence of our bound is suboptimal, it is widely anticipated that the random XXZ spin chain should still exhibit some degree of slow propagation (rather than none) as discussed, for example, in \cite{abanin2019}.  Nevertheless, in the specific energy interval known as the droplet spectrum  the time dependence can be suppressed completely  \cite{EKS3}.

\subsection{Comparison with existing results}
Up until a few years ago, the landscape of rigorous MBL-related results  was primarily confined to a small class of exactly solvable models (see \cite{ARNSS} for a comprehensive review). The latter are not expected to capture the full complexity and generality of the phenomenon in realistic disordered many-body systems.

More recently, progress has been made for the random XXZ  spin chain, a paradigmatic system often used to study MBL. Specifically, spectral and dynamical localization were rigorously established within the particular energy interval $[\fd, 2\tfd)$, termed the {\it droplet spectrum} \cite{BeW,EKS1}. Subsequent work has rigorously demonstrated other defining characteristics of MBL within this energy window. These include the essential features of quantum information scrambling and entanglement scaling in MBL systems: the exponential clustering properties for the associated many-body eigenfunctions,  the non-propagation of information \cite{EKS3}, and the area law for  the entanglement entropy \cite{BeW1}. 

The authors' subsequent work showed the persistence of spectral and eigenstate localization in the random XXZ spin chain for any fixed energy interval at the bottom of the spectrum. This was first  proved for finite systems with volume-dependent estimates \cite{EK22}, and subsequently for infinite systems \cite{EK25}. The latter paper also provides a proof of a weak form of dynamical localization, showing the exponential decay of eigencorrelators in Hausdorff distance. 
While the analog of this estimate implies non-spreading of wave packets for a single-particle system, its utility in establishing  the logarithmic light cone feature of MBL   remains  unclear.

 For finite systems, we  established this MBL feature  for any fixed energy interval   in \cite{EK24}.  Crucially, however, the volume-dependent bounds in \cite{EK24} prevented us from drawing conclusions about the system's behavior in the thermodynamic limit.  In the current paper, we overcome this limitation by completely removing the volume dependence, thereby proving slow propagation of information in infinite systems.
To do so, we devise a new way to approximate the infinite system by finite systems, overcoming the difficulties inherent in the tensor product nature of the Hilbert spaces.   We expect this approximation  to have applications in other contexts as well.

Finally, we also want to mention the works  \cite{DGHP} \and \cite{Lemm}. The former shows anomalous heat transport at high disorder in a chain defined by a diagonal random Hamiltonian  subjected to a local perturbation. The latter examines $d$-dimensional spin systems, also defined by a diagonal Hamiltonian (satisfying a certain non-resonance condition)  subjected to a local perturbation, and derives a Lieb-Robinson bound that couples the perturbation size to the time scale.

\subsection{Sketch of the proof of Theorem \ref{thm:main}}

Our starting point is  \cite[Theorem 2.6]{EK24}, a version of  Theorem \ref{thm:main} for the finite volume random XXZ spin chain.  (We state the version of  \cite[Theorem 2.6]{EK24} we use as Theorem \ref{thm:localmodell}     below.)  The main difficulty in  applying this theorem is that  the   estimate corresponding  to  \eq{eq:main}  includes a factor with polynomial dependence on the volume.

We prove a more detailed version of Theorem \ref{thm:main}, namely Theorem \ref{thm:main2}.  The main step in the proof consists in showing  that, given an initial observable $T$ supported by a finite interval $\cX$,   in the presence of randomness one can approximate, in expectation, the restriction  $ P_{E}\e^{itH}T\e^{-itH}P_{E}$ of the Heisenberg evolution  to the energy window $ [0,E]$ by $P^\La_{E}\e^{itH^\La}T\e^{-itH^\La}P^\La_{E}$, where  $H^\La$ is the restriction of $H$ to an appropriately chosen finite interval $\La \subset \Z$, and $P^\La_{E}=\chi_{ [0,E]}(H^\La)$. 

 Typically, the expression $f(H)T\bar f(H)$ cannot be approximated by $f(H^\La)T\bar f(H^\La)$ (as is evident from  taking $f(x)=x$). However, in our specific context, this approximation is possible when $f(x)= \chi_{[0,E]}(x)e^{itx}$  and the size of $\La$ is  proportional to $|t|$. The bulk of the work in this paper is dedicated to establishing this key result.

Once the effective size of the system is reduced to the linear light cone, we invoke our previous result for slow propagation of information in a finite volume, as stated in Theorem \ref{thm:localmodell}.  The polynomial dependence on the  volume of  \cite[Theorem 2.6]{EK24} is then absorbed into the polynomial dependence on time to produce Theorem \ref{thm:main}.

\section{The main result revisited}

\subsection{Notation and definitions}

 Given $\La \subset \Z$, the canonical (orthonormal) basis $\Phi_\La  $ for $\cH_\La$ is given by
\be \label{eq:standbasis}
& \Phi_\La 
=\set{\phi^\La_A=\pa{\prod_{i\in A}\sigma_i^-}\phi^\La_\emptyset:\ A\subset \Lambda, \ \abs{A} <\infty},
 \ee
where $\phi^\La_\emptyset= \bigotimes_{i\in \La}\uparrow\rangle_i$, and $\abs{A}$ denotes the cardinality of a set $A\subset \Z$.

The total number of particles (in $\La$) operator is given by $\cN^\La= \sum_{i\in \La} \cN_i$. It is diagonalized in the canonical  basis:   $\cN^\La \phi^\La_A= \abs{A} \phi^\La_A$.

  If ${\Lambda}\ne \emptyset$, 
we  define the orthogonal projections  $P_\pm^{\Lambda}$ on  $\cH_{\Lambda}$ by 
 \begin{align}\label{P+-S}
P_+^{\Lambda}=\bigotimes_{i\in {\Lambda}}\pa{{\ident}_{\mathcal H_{i}}-\cN_i}= \chi_{\set{0}}\pa{\cN^{\Lambda}}  \qtx{and} 
P_-^{\Lambda}= {\ident}_{\mathcal H_{\Lambda}}-P_+^{\Lambda}=  \chi_{[1,\infty)}\pa{\cN^{\Lambda}}.
\end{align} 
$P_+^{\Lambda}$  is the orthogonal projection onto states with no particles in the set ${\Lambda}$;  $P_-^{{\Lambda}}$  is the orthogonal projection onto states with at least one particle in ${\Lambda}$. (Note that $P_-^{\set{i}}=\cN_i$ for $i\in  \Z$.)
We also set 
 \be \label{P+-empty}
 P_+^{\emptyset } ={\ident}_{\cH_{\Lambda}} \qtx{and } P_-^{\emptyset }=0.
 \ee
  
We consider $\Z$ as a graph with the usual graph  distance  $d_\Z(i,j)= \abs{i-j}$ for  $i,j\in \Z$.  We will  consider $\La \subset \Z$ as a subgraph of $\Z$,  and
 denote by $\dist_\La( \cdot,\cdot) $ the  graph distance in $\La$, which can  attain the infinite value   if $\La$ is not a connected subset of $\Z$.
 
  Given  $ {A}\subset \Lambda\subset \Z$,   we set
   \be
{[{A}]^\La_s } &=\begin{cases}\set{i\in\La: \dist_\La\pa{i,{A}}\le s} &\sqtx{if}  s\in \N^0=\set{0} \cup \N  \\ 
\set{i\in\La:\dist_\La \pa{i,  \La \setminus {A}}\ge 1-s}= {A}\setminus [\La \setminus {A}]^\La_{-s}&\sqtx{if}  s\in-\N
\end{cases}.
\ee
If ${A}=\set{j}$ we write ${[j]^\La_q }={[\set{j}]^\La_q }$.
For $s\in \N$,  we also set
\be
\partial_{s}^{\La, out} {A}&=\set{i\in\Lambda:\ 1\le \dist_\La\pa{i,{A}}\le s}= [{A}]^\La_s\setminus {A},\\
 \partial_{s}^{\La,in} {A}&=\set{i\in\Lambda:\ 1\le  \dist_\La \pa{i,\La \setminus {A}}\le s}={A}\setminus [{A}]^\La_{-s},\\
 \partial_s^\La {A} &=  \partial_{s}^{\La, in} {A} \cup  \partial_{ex}^{\La, out}{A}=[{A}]^\La_s \setminus  [{A}]^\La_{-s},\\
{\bpartial}^\La {A} &=  \set{\set{i,j}\subset \La: \  (i,j) \in \pa{\partial_{1}^{\La,in} {A} \times \partial_{1}^{\La,out}{A} } \cup \pa{\partial_{1}^{\La,out} {A} \times \partial_{1}^{\La,in}{A} }}.
 \ee
If $s=1$, we usually omit it from the notation altogether.

Given bounded operators $S$ and $Y$ on a Hilbert space, we use the notation $\pa{S}_Y=YSY$. 

When $\La=\Z$ it will be mostly omitted from the  notation.

\subsection{The XXZ Hamiltonian on subsets of \texorpdfstring{$\Z$}{Z}}
 We fix $\Delta >1$ and $\lambda >0$. Since $\sigma^z=\ident -2\cN$, we can rewrite $H_0$  as  
\be\label{H0h}
H_0 =  \sum_{i\in\Z} h_{i,i+1} +  \cN^\Z=  \sum_{\set{i,i+1}\subset \Z} h_{i,i+1} +  \cN^\Z ,
\ee  
where
\be
 h_{i,i+1} &=  -\tfrac{1}{2\Delta}\pa{\sigma_i^+\sigma_{i+1}^-+\sigma_i^-\sigma_{i+1}^+} - \mathcal{N}_{i}\mathcal{N}_{i+1}.
\ee
Note that $h_{i,i+1}$ is an operator on $\cH_i\otimes \cH_{i+1}\cong \C^4$ with 
\be\label{hii+1}
\norm{h_{i,i+1} }=1.
\ee

This representation of $H_0$ motivates the following definition  of the random XXZ Hamiltonian $H^\La$ on a subset  $\La \subset \Z$:
\be
H^\La=H^\La_{\bmo}= H_0^\La + \lambda V_{\bmo}^\La 
\qtx{acting on} \cH_\La,
\ee  
where
\be
H_0^\La =   \sum_{\set{i,i+1}\subset \La} h_{i,i+1} +  \cN^\La , 
\ee
and
\be
V^\La_{\bmo}=V^\La_{{\bmo}_\La} = \sum_{i\in\La} \omega_i \mathcal{N}_i,   \quad {\bmo}_\La=\set{\omega_i}_i\in \La.
\ee
This formulation incorporates a boundary condition if $\La \ne \Z$.  (If $\La=\Z$ we mostly omit the subscript or superscript $\Z$.)

We set  $R_z^\La= (H^\La -z)^{-1}$   for $z \not\in \sigma(H^\La)$, the resolvent of   $H^\La $.

The  free Hamiltonian $H_0^\Lambda$ can be rewritten as
\beq\label{eq:modH}
H_0^\Lambda=-\tfrac{1}{2\Delta} \bD^\Lambda + \cW^\Lambda \,
\eeq
where 
\begin{align}\label{bD}
 \bD^\Lambda = \sum_{\set{i,i+1}\subset \Lambda} \pa{\sigma_i^+\sigma_{i+1}^-+\sigma_i^-\sigma_{i+1}^+} \qtx{and}
 \cW^\Lambda = \cN^\Lambda -\sum_{\set{i,i+1}\subset \Lambda} \cN_i\cN_{i+1} .
 \end{align} 
 
The operators     $\cW^\La$, and $V^\Lambda_\omega$  are diagonalized in the canonical (orthonormal) basis $\Phi_\La  $ for $\H_\La$. More precisely, for all $A\subset \La$ with $\abs{A}<\infty$ we have:
 \begin{itemize}
 
 \item   $V^\Lambda_{\bmo}   \phi_A^\La= \omega\up{A}  \phi_A^\La$,  where $\omega\up{A}=\sum_{i\in A} \omega_i$,  

 \item   $\cW^\La \ \phi_A^\La= W^\La_A \phi_A^\La$,
where  $W^\La_A\in [0,\abs{A}]\cap \N^0$ is the number of clusters of $A$ in $\La$, i.e., the number of connected components of  $A$ in $\Lambda$ (considered as a subgraph of $\Z$).   
\end{itemize}

  An important feature of the XXZ Hamiltonian  $H^\Lambda$  (and $H_0^\La$) is  particle number conservation: $H^\La$ and  $ \cN^\La$ commute, that is,  all bounded functions of $H^\La$ and $ \cN^\La$ commute.

It can be verified (e.g., \cite{EK22})    that 
\beq\label{H0W}
0\le \tfd \cW^\Lambda\le H_{0}^\Lambda \le H^\Lambda,
\eeq
$H^\La \phi^\La_\emptyset= H^\La_0 \phi^\La_\emptyset=\cW^\Lambda \phi^\La_\emptyset=0$, and $0$ is a simple eigenvalue for these operators.  It follows that the spectrum of $H^\Lambda$ is  of the form  
\beq
\sigma(H^\Lambda)=\set{0} \cup \pa{\left[1 -\tfrac 1 \Delta, \infty \right ) \cap  \sigma(H^\Lambda) }.
\eeq

The lower bound in \eq{H0W}  leads to the introduction of the energy thresholds $k\tfd$, $ k=0, 1,2\ldots$. 
For technical reasons we introduce  energy intervals
labeled by $ q \in \tfrac 12 \Z$ as in \cite{EK24}:
 \be\label{Ikle2}
I_{\le q}&= \left(-\infty, (q+\tfrac 3 4)\tfd\right] ,
\quad 
I_q= \left[\sfd, (q+\tfrac 3 4)\tfd\right].
\ee

Given $K\subset \La\subset \Z$, where  $K$ is a finite  interval (possibly $K=\emptyset$) and $\La$ is nonempty (possibly $\La=\Z$), 
 we  set 
 \be\label{HdaggerK}
 H^{\La\upharpoonright  K}&= H^K + H^{\La \setminus K} \qtx{and} \Gamma^{\La{\upharpoonright} K}= H^\La-H^{\La{\upharpoonright} K}\qtx{acting on} \cH_\La,\\
H^{{\upharpoonright} K}&=H^{\Z{\upharpoonright} K}, \quad\Gamma^{{\upharpoonright} K} =\Gamma^{\Z{\upharpoonright} K}.
\ee
(Note that $H^{\La{\upharpoonright} \emptyset}= H^\La $, $H=H^\Z$.)
  Since $K$ is an interval, it follows from \eq{hii+1} that
\be\label{GammaK}
\norm{\Gamma^{\La{\upharpoonright} K}}  \le 2.
\ee

We also set
\be
R^{\La{\upharpoonright} K}_z&= \pa{H^{{\upharpoonright} K} -z}^{-1}\qtx{for} z\in \C \setminus \sigma\pa{H^{{\upharpoonright} K} },\\
P^{\La{\upharpoonright} K}_{I_{\le q}}& = \chi_{I_{\le q}}(H^{\La{\upharpoonright} K})=P^{\La{\upharpoonright} K}_{(q+\frac 3 4)\nfd }, \quad P^{\La{\upharpoonright} K}_{I_{q}}= \chi_{I_{ q}}(H^{\La{\upharpoonright} K}),\\
R^{{\upharpoonright} K}_z  &  =R^{\Z{\upharpoonright} K}_z, \quad P^{{\upharpoonright} K}_{I_{\le q}} =P^{\Z{\upharpoonright} K}_{I_{\le q}}, \quad P^{{\upharpoonright} K}_{I_{q}} =P^{\Z{\upharpoonright} K}_{I_{q}}.
\ee

The evolution of an observable $T$ on $\La$ by the Hamiltonian $H^{\La{\upharpoonright} K}$ is given by
\be
\tau_t^{\La{\upharpoonright} K}(T)= \e^{itH^{\La{\upharpoonright} K}} T  \e^{-itH^{\La{\upharpoonright} K}} \qtx{for} t \in \R.
\ee

We  set $ \tau_t^{{\upharpoonright} K}=\tau_t^{\Z{\upharpoonright} K}$.  If $K=\emptyset$ we omit ${\upharpoonright} K$ from the notation.

We observe that
\be\label{suppsK}
 \tau_t^{\La{\upharpoonright} K}(T)= \tau_t^{K}(S) \qtx{if} \supp S \subset K,
  \ee
and
\be\label{PLaKK}
P^{\La{\upharpoonright} K}_{I\le q}= P^{\La{\upharpoonright} K}_{I\le q}P^{K}_{I\le q} = P^{K}_{I\le q} P^{\La{\upharpoonright} K}_{I\le q}\qtx{for} q\in \tfrac 12 \N^0.
\ee

\subsection{Finite volume theorem and detailed main result}

Given $ q\in  \tfrac 12\N^0$, we   define $\alpha_q=13+\beta_{q+\frac 12} \in \N^0$, where  $\beta_q $ is defined recursively  for $ q\in  \tfrac 12\N^0$
  by   $\beta_0=0$, $ \beta_q=\beta_{q-\frac 12} +9\cl{q }+ 13$ for $q\in \frac 12 \N$.
Note that  $ 9 q^2 + \frac {61}2 q  \le \beta_q\le 9 q^2 +  \frac {79}2 q $.  
 
We now state a version of   \cite[Theorem 2.6]{EK24}.

\begin{theorem}[ Slow propagation of information-finite volume]\label{thm:localmodell}
 Fix parameters $\Delta_0>1$ and  $ \lambda_0 >0$.   Then for all   $q\in \frac 12 \N^0$  there exist strictly positive constants $ C_q, Y_q,   \gamma_q, \rho_q,\theta_q$ (depending on $q$, $\Delta_0$, $ \lambda_0$) such that for all $\Delta \ge \Delta_0$ and $\lambda \ge \lambda_0$ with  $\lambda \Delta^2\ge Y_q$ the following assertion holds: 
Given an observable $T$  supported on a finite  interval $\mathcal X\subset \Z$ with $\norm{T}\le 1$,  a scale $\ell\in \N$,  and a finite interval   $\La \subset \Z$  such that  $ [\cX]_{\alpha_q\ell}\subset \La$,   for all $t\in \R$
there exists a random observable ${T}_{q,t,\ell}={T}(\bmo,q, t,\ell)$, supported on  $ [\cX]_{\alpha_q\ell}$ and independent of $\La$,  such that
\be\label{T_t}
\norm{ {T}_{q,t,\ell}} \le C_q \langle t\rangle^{\gamma_q},
\ee
and
\be\label{eq:locality2mgq}
 \E\norm{P^\Z_{I_{\le q}}\pa{\tau^\La_t(T)-  {T}_{q,t,\ell}} P^\Z_{I_{\le q}}} \le   C_{q} \langle t\rangle^{\gamma_q}\abs{\Lambda}^{{\rho_{q} }}\e^{- \theta_{q}\, \ell}.
\ee
\end{theorem}

\begin{remark}
Theorem \ref{thm:localmodell} as stated follows from  \cite[Theorem 2.6]{EK24} and its proof. 
Although \cite[Theorem 2.6]{EK24} does not explicitly state the independence of  ${T}_{q,t,\ell}$  from   $\La  \supset[\cX]_{\alpha_q\ell}$  and the bound \eqref{T_t}, 
these properties can be extracted from  the construction of  ${T}_{q,t,\ell}$   in the proof of   \cite[Theorem 2.6]{EK24}. 
\end{remark}

We will prove the following theorem, which implies Theorem \ref{thm:main}.   Given $q\in \frac 12 \N^0$,   $ Y_q,   \gamma_q, \rho_q,\theta_q$ will denote the constants
from Theorem \ref{thm:localmodell}.    ($C_q$ will always be a generic notation for a constant depending on $q$.) In addition, Given $q\in \frac 12 \N^0$, a scale $\ell\in \N$,  $t \in \R$, and an observable $T$  supported on a finite  interval $\mathcal X\subset \Z$ with $\norm{T}\le 1$, by  ${{T}_{q,t,\ell}}={T}(\bmo, t,q,\ell)$ we denote the random observable constructed in Theorem \ref{thm:localmodell}, supported on $ [\cX]_{\alpha_q\ell}$, and  satisfying 
\eq{T_t} and \eq{eq:locality2mgq} for $\La  \supset[\cX]_{\alpha_q\ell}$.

\begin{theorem}[ Slow propagation of information-infinite volume,  detailed version]\label{thm:main2}
 Fix parameters $\Delta_0>1$ and  $ \lambda_0 >0$.   Then for all   $q\in \frac 12 \N^0$  there exist strictly positive constants  $\tilde \gamma_q$ and $\tilde\theta_q$ (depending on $q$, $\Delta_0$, $ \lambda_0$) such that for all $\Delta \ge \Delta_0$ and $\lambda \ge \lambda_0$ with  $\lambda \Delta^2\ge Y_{q+\frac 12}$ the following assertion holds: 
Given an observable $T$  supported on a finite  interval $\mathcal X\subset \Z$ with $\norm{T}\le 1$ and a scale $\ell\in \N$, the following holds  for all $t\in \R$:
\be\label{eq:main2}
    \E \norm{P^\Z_{I_{\le q}}\pa{\tau^\Z_t(T)- {T}_{q+\frac 1 2,t,\ell}}P^\Z_{I_{\le q}}}\le C_{q,\abs{\cX}} \langle t\rangle^{\tilde \gamma_{q}}e^{-\tilde \theta_{q}\, \ell}.
\ee
 \end{theorem}   
 Throughout the paper, we will use generic  nonnegative constants $C, c,m$, etc., whose values will be allowed to change from line to line, or even in the same line  in a displayed equation. These constants will not depend on subsets of $\Z$, but    they will,  in general,  depend on  $\mu$, $\Delta_0$, and $ \lambda_0$.  (We will  always assume $ \Delta \ge \Delta_0>1$ and  $\lambda\ge \lambda_0>0$. In the proofs we may get  constants depending on $\Delta$ and $\lambda$,  but with minor modifications we can get constants depending uniformly on $ \Delta \ge \Delta_0$ and  $\lambda\ge \lambda_0$.) This dependence will not be made explicit, except occasionally.  We will indicate the dependence of a constant on  $k\in \N^0$ or $q\in \frac 12 \N^0$ explicitly by writing it as $C_k, m_k, C_q, c_q$, etc. These constants can always be estimated from the arguments, but we will not track the changes to avoid complicating the arguments.

\section{ Guiding principles  for the proof of the main result}

The proof of Theorem \ref{thm:main2}  is built upon three core guiding principles, which are formulated and proven in this section. Loosely speaking, they are:

{\bf Restriction on the location of particles}: The Fermi projection $P^\Z_{I_{\le k}}$, $k\in \N$, restricts the possible 
configurations of particles (spin-down states) within any interval $\La \subset \Z$:  the number of groups of particles located  inside $\La$ that are separated by scales comparable with or exceeding $\ln \abs{\La}$ is of order $k$,  outside an event of small probability.
The precise statement is given in  Lemma  \ref{modLemma}. The proof relies on Proposition  \ref{lemEkR}, which follows from  the quasi-locality property for the resolvent and large deviation estimates.

{\bf Approximation of an  energy interval  indicator}: The indicator function for a fixed energy interval can be exponentially well approximated by functions whose Fourier  transforms have compact support. This is shown in Lemma \ref{lemfxi}.

{\bf Finite speed of propagation}: The propagator $\e^{itH}$ cannot flip a contiguous block of $L$ up-spins in a time $t$ shorter than $cL$ for some $c>0$. The precise statement,  Lemma \ref{lemB1}, is a consequence of the finite speed of propagation inherent in local spin chains and the specific structure of the XXZ Hamiltonian.

\subsection{Restriction on the location of particles}

 For a given  $m \in \N^0$, we set $Q_m^\Lambda=\chi_{\set{m}}\pa{\cW^\Lambda}$, the orthogonal projection in $\H_\La$ onto  configurations  with exactly $m$ clusters, and let  $Q_B^\Lambda=\chi_{B}\pa{\cW^\Lambda}=\sum_{m\in B} Q_m^\Lambda$  for $ B\subset \N^0$. Note that  $Q_0^\Lambda=  P_+^\Lambda$.
 For $k\in \N$, we set 
\be\label{QkhatQ}
Q_{\le k}^\Lambda   =Q_{\set{1,2,\ldots,k}}^\Lambda =\sum_{ m=1}^k Q_m^\Lambda \qtx{and} 
\what Q_{\le k}^\Lambda   =Q_{\le k}^\Lambda + \tfrac {k+1} k Q_0^\Lambda.
\ee

We also set  
\be \label{eq:compH'}
\what  H_0^{ \Lambda}&=H^{\Lambda}+\tfd Q_0^\Lambda,\\
\what  H_k^{ \Lambda}&=H^{\Lambda}+{k}\tfd \what  Q_{\le k}^{\Lambda} \qtx{for}  k\in \N,\\
\what  R^{\Lambda}_{k,z}&= \pa{\what  H_k^\La  -z}^{-1}  \mqtx{for} z \notin \sigma(\what  H_k^\La), \ k\in \N^0.
\ee 
It follows  that   for $k\in \N^0$ we have (see \cite[Eq. (3.21)]{EK22})
 \be \label{eq:hatH1'}
\what  H_k^{ \Lambda}\ge  \pa{k+1} \tfd    \qtx{and}  \pa{\what  H_k^{ \Lambda}-E}  \ge \tfrac 1{4}\tfd  \mqtx{for} E \in I_{\le k}.
\ee

We fix $\Delta_0 \ge 5$ (the case $1<\Delta <5$ can be handled by a modification of the argument as discussed in \cite[Remark 3.3]{EK22}) 
and  $\lambda_0 >0$, and assume $\Delta \ge \Delta_0$ and $ \lambda \ge \lambda_0$.

\begin{proposition}\label{lemEkR}
Given a finite interval $\La \subset \Z$, $R\in \N$, and $k\in \N$, let 
 $Y\up{R}_-=\chi_{[0,R]}(\mathcal N^\La)$ and  $Y\up{R}_+=\ident - Y\up{R}_-$, and   consider the  random operator on $\cH_\Z$ given by
\be
\tilde H^Z_{k,R}=H^\Z+(k+4)\tfd Y\up{R}_-\what  Q_{\le k+4}^{Z}.   
\ee
Then for all intervals $A\subset \Z$ and  $r\in \N$ we have 
\be\label{tildeRkR}
\norm{P_-^{A}\tilde R_{k,R,z}^\Z P_+^{[A]^\Z_{r}}}\le \tfrac 1 {\eta_z}\e^{-\pa{\ln \Delta \eta_z}  r} \sqtx{for} z\in \C \qtx{with} \eta_z=\dist \pa{z, \sigma \pa{\tilde H^\Z_{k,R}}}>0,
\ee
where $\tilde R_{k,R,z}^\Z=(\tilde H^\Z_{k,R}  -z)^{-1}$.

Moreover,  there exists an event  $\cE_{k,R}$,   satisfying
\be\label{PEkR}
\P\set{\cE_{k,R}} \ge 1- C_{k} C_\mu \abs{\La}^{2k+15} \e^{-c_\mu R},
\ee
 such that  for $\bmo \in \cE_{k,R}$ we have
\be\label{eq:tildeH}
\tilde H^\Z_{k,R}\ge (k+1) \tfd .
\ee
\end{proposition}

\begin{proof}

The estimate \eq{tildeRkR} follows immediately from \cite[Lemma 3.1 and Remark 3.2]{EK22}.

 Since $[Y\up{R}_-,\what  Q_{\le k+4}^{Z}]=0$, 
 it follows from \eq{eq:hatH1'} that
\be 
Y\up{R}_-\tilde H^\Z_{k,R} Y\up{R}_-  & =Y\up{R}_-\pa{H^\Z+(k+4)\tfd\what  Q_{\le k+4}^{Z}}Y\up{R}_-
\\ & =Y\up{R}_- \what H^\Z_{k+4} Y\up{R}_-  \ge (k+5)\tfd Y\up{R}_-.
\ee

 We also have (here $\La^c=\Z \setminus \La$, $\Gamma^{\La}=\Gamma^{\Z,\La}$)
\be 
Y\up{R}_+\tilde H^\Z_{k,R}  Y\up{R}_+   & =Y\up{R}_+H^\Z Y\up{R}_+=Y\up{R}_+H^\La Y\up{R}_+  +Y\up{R}_+H^{\La^c} +Y\up{R}_+\Gamma^{\La} Y\up{R}_+ \\
& \ge Y\up{R}_+H^\La Y\up{R}_+ +Y\up{R}_+\Gamma^\La Y\up{R}_+
 \ge  Y\up{R}_+H^\La Y\up{R}_+ -2 Y\up{R}_+ \\
 &\ge  Y\up{R}_+H^\La Y\up{R}_+ -4 \tfd Y\up{R}_+,
\ee
where we used \eq{GammaK} and $\Delta \ge \Delta_0 > 2$.

We now use a large deviation estimate as in \cite[Proof of Lemma 3.7]{EK22}. It follows from 
 \cite[Eqs. (3.52)--(3.55)]{EK22} that for all $p\in \N$ and  $N\in \N$ there exists an event  $\cE_{p}\up{N}$ such that
 \be
 \P \set{\cE_{p}\up{N}} \ge 1 - C_p \abs{\La}^{2p}\e^{-c_\mu N}
 \ee
 and
\be
H^\La \chi_{\set{N}}(\mathcal N^\La) \ge (p+1)\tfd\chi_{\set{N}}(\mathcal N^\La)  \qtx{for} \bmo \in \cE_{p}\up{N}.
\ee
Thus, setting $\cE_{k,R}= \bigcap_{N=R+1}^{\abs{\La}}\cE_{k+9}\up{N}$, we have
\be
\P\set{\cE_{k,R}} \ge 1-  C_{k+9} \abs{\La}^{2k+18}\sum_{N=R+1}^{\abs{\La}}  \e^{-c_\mu N}\ge  
1 - C_{k} C_\mu \abs{\La}^{2k+15} \e^{-c_\mu R},
\ee
and
\be
Y\up{R}_+H^\La  \ge (k+9) \tfd Y\up{R}_+ \qtx{for} \bmo \in \cE_{k,R}.
\ee

SInce  we have
\be
\norm{Y\up{R}_-\Gamma^\La Y\up{R}_++Y\up{R}_+\Gamma^\La Y\up{R}_-}\le \norm{\Gamma^\La}\le 2 \le 4\tfd,
\ee
as
\be
&\norm{Y\up{R}_-\Gamma^\La Y\up{R}_++Y\up{R}_+\Gamma^\La Y\up{R}_-}^2= \norm{\pa{Y\up{R}_-\Gamma^\La Y\up{R}_++Y\up{R}_+\Gamma^\La Y\up{R}_-}^2} \\
& \quad = \norm{Y\up{R}_-\Gamma^\La Y\up{R}_+\Gamma^\La Y\up{R}_-  +Y\up{R}_+ \Gamma^\La Y\up{R}_-\Gamma^\La Y\up{R}_+ } \\
& \quad =\max \set{\norm{Y\up{R}_-\Gamma^\La Y\up{R}_+\Gamma^\La Y\up{R}_- },\norm{Y\up{R}_+\Gamma^\La Y\up{R}_-\Gamma^\La Y\up{R}_+ } }\\
&\quad  
\le \max \set{\norm{\Gamma^\La Y\up{R}_+\Gamma^\La },\norm{\Gamma^\La Y\up{R}_-\Gamma ^\La} }  \le \norm{\Gamma^\La}^2,
\ee
we conclude that for $\bmo \in \cE_{k,R}$  we have 
\be\label{eq:tildeH2}
\tilde H^\Z_{k,R}  &\ge   (k+5)\tfd Y\up{R}_-  +  (k+5)\tfd Y\up{R}_+  -4 \tfd  \\
&  =  (k+5)\tfd  -4\tfd  = (k+1) \tfd.
\ee
\end{proof}

We recall that given $S\subset \Z$,  $W_S\in  \N^0$ is the number of clusters of $S$ in $\La$, i.e., the number of connected components of  $S$.
\begin{lemma}\label{modLemma}
Let $\La \subset \Z$ be a finite interval and $k\in \N$.  Given   a collection   $\set{S_i}_{i=1}^{k+5}$ of  nonempty subsets   of  $\Lambda$ with
\beq\label{dSiSj}
\min_{i\neq j }\dist\pa{S_i,S_j}\ge 2{L} +1, \qtx{where}{L} \in \N,
\eeq 
 we have
\be\label{PQPLD}
\E \Big\|P_{ I_{\le  k}}\prod_{i=1}^{k+5} P_-^{S_i}\Big\|\le  C_{\mu,k} W_{max}\abs{\La}^{2k+15} \e^{-c^\pr_{\mu}{L}},
\ee
 where $W_{max} =\max_{i=1}^{k+5} {W}_{{S_i}}$.
\end{lemma}

\begin{proof} 
Note that  $P_{ I_{\le  k}}\prod_{i=1}^{k+5} P_-^{S_i}=P_{ I_{k}}\prod_{i=1}^{k+5} P_-^{S_i}$.
We can represent $P_{ I_{k}}$ as a contour integral
\be
P_{ I_{ k}}=\frac1{2\pi i}P_{ I_{ k}}\oint_\Upsilon R_{z}^{\Z}dz,
\ee
where  the contour $\Upsilon$ is defined by $\Upsilon=\set{z\in\C:\ \min_{x\in  I_{ k}}|x-z|=\frac18\tfd}$. (Note that  $\norm{P_{ I_{k}}R_{z}^{\Z}}\le \frac  8{\fd}$ for any $z\in\Upsilon$.)

Let $R=2 {L} +k+5$ and consider the event $\cE_{k,R}$ of Proposition \ref{lemEkR}.  Let $\bmo \in \cE_{k,R}$.
Using the resolvent identity    and \eqref{eq:tildeH}, we deduce that
\be
P_{ I_{ k}}=\frac{(k+4)\tfd }{2\pi i}P_{ I_{ k}}\oint_\Upsilon R^{\Z}_{z} Y_-\up{R} \what Q^\Z_{\le k+4}\tilde R^{\Z}_{k,R,z}dz.
\ee
Using  \eq{dSiSj}  (recall $\La$ is an  interval) and the choice of $R$,  we have
\beq\label{ThetaPPP}
Y\up{R}_- \prod_{i=1}^{k+5} P_-^{[S_i]^\Z_{{L}}} =Y\up{R}_- Q^\La_{>k+4 }\prod_{i=1}^{k+5}  P_-^{[S_i]^\Z_{{L}}}.
\eeq
Moreover, \eqref{eq:tildeH} and \eq{tildeRkR} yield
 \be\label{P+tildeRP}
\Big\|P_+^{[S_i]^\Z_{{L}}}\tilde  R^{\Z}_{k,R,z}P_-^{S_i}\Big\|\le C W_{{S_i}} \e^{-m_0{L}}\qtx{for} i=1,2,\ldots k+5,
\ee
where   $C=\frac  8  {\fd}$ and $m_0= \ln  \frac {\Delta-1} {8}$.
Using $P_+^{[S_i]^\Z_{{L}}}+ P_-^{[S_i]^Z_{{L}}}= {\ident}_{\cH_\Z}$ and \eq{P+tildeRP}, we have that
\be\label{PQPi2}
\Big\|{ Y_-\up{R} \what Q^\Z_{\le k+4}\tilde R^{\Z}_{k,R,z}\prod_{i=1}^{k+5} P_-^{S_i}}\Big\| 
&  \le C (k+5) W_{max}\e^{-m_0{L}}+\Big\|{ Q^Z_{\le k+4}Q_{>k +4}^\La Y_-\up{R} \prod_{i=1}^{k+5}  P_-^{[S_i]^\La_{{L}}} \what  R^{\Lambda}_{k,z}}\Big\| \\
&=  { C} (k+5)W_{max} \e^{-m_0{L}},
\ee 
where we used \eq{ThetaPPP}  and  $ Q^\Z_{\le k+4}Q_{>k+4 }^\La=0$.

Since $ \Big\|{P_{ I_{\le  k}}\prod_{i=1}^{k+5} P_-^{S_i}}\Big\|\le 1$, it follows, using \eq{PEkR} as well, that
\be\label{PPiP-}
\E \Big\|{P_{ I_{\le  k}}\prod_{i=1}^{k+  5} P_-^{S_i}}\Big\|&\le { C} (k+5)W_{max}\e^{-m_0{L}}+   C_{k} C_\mu \abs{\La}^{2k+15} \e^{-c_\mu  {R}}\\
& \le C_{\mu,k}W_{max}\abs{\La}^{2k+15} \e^{-c^\pr_{\mu}{L}}.
\ee
\end{proof}
 
 \subsection{Approximation of an  energy interval  indicator}

Given a  finite close interval $I\subset \R$, we set   $[I]_s=\set{t\in \R; \ \dist (t,I)\le s}$ for $s\ge 0$.

\begin{lemma}\label{lemfxi}
 Let $I\subset \R$ be a finite closed interval, $\xi>0$, $\theta >0$, $\Phi_\xi(x)=\sqrt{\frac{ \xi}{\pi }}\e^{-\xi x^2}$, and let $f_{\xi,\theta}=f_{\xi,\theta}\up{I}$ be the convolution of $\chi_{[I]_{2\theta}}$ and $\Phi_\xi$, 
\be\label{fxith}
f_{\xi,\theta}=f_{\xi,\theta}\up{I}=\chi_{[I]_{2\theta}} \ast \Phi_\xi.
\ee
Then
\begin{enumerate}
\item  We have
\be
&0\le f_{\xi,\theta} (x) \le 1\sqtx{for all}x\in\R,\\
&1-\tfrac 2{\sqrt{\pi \xi }\theta}\e^{-\xi \theta^2}\le f_{\xi,\theta}(x) \le 1   \sqtx{for all}x\in I, \qtx{so}\chi_I \le  
\pa{1-\tfrac 2{\sqrt{\pi \xi }\theta}\e^{-\xi \theta^2}}^{-1} \chi_I f_{\xi,\theta}, \\
&   f_{\xi,\theta}(x)\le \tfrac 2{\sqrt{\pi \xi }\theta}\e^{-\xi \theta^2}\sqtx{for all} x\in \R \setminus[I]_{3\theta}, \qtx{so}
 f_{\xi,\theta} \le   \chi_{[I]_{3\theta} } + \tfrac 2{\sqrt{\pi \xi }\theta}\e^{-\xi \theta^2} \chi_{\R \setminus [I]_{3\theta} }.
\ee

\item The Fourier transform $\hat  f_{\xi,\theta} (s)= \frac {1}{\sqrt{2\pi}}\int_{\R} f_{\xi,\theta} (x)\e^{isx}\, \d x$ satisfies
\be\label{hatfxitheta}
 \|\hat  f_{\xi,\theta} \|_1&\le   \sqrt{2\xi} (|I|+4\theta),\\
\int_{\R \setminus [-\zeta,\zeta]}\abs{\hat f_{\xi,\theta} (s)} \d s&\le  2 \sqrt {\tfrac 2 \pi}  (|I|+4\theta)\tfrac \xi \zeta \e^{- \frac {\zeta^2}{4\xi}} \qtx{for} \zeta > 0.
\ee

\item Let  
\be\label{deftildef}
\tilde   f_{\xi,\theta,\zeta}(x)=  \frac 1 {\sqrt{2\pi} }\int_{[-\zeta,\zeta]}\hat   f_{\xi,\theta}(s)\e^{-i sx} \d s,  \qtx{where}\zeta>0.
\ee
Then
\be\label{eq:fdiff5}
\norm{\tilde   f_{\xi,\theta,\zeta}-   f_{\xi,\theta}}_\infty\le   2 \sqrt {\tfrac 2 \pi}  (|I|+4\theta)\tfrac \xi \zeta \e^{- \frac {\zeta^2}{4\xi}}
\ee

\end{enumerate}
\end{lemma}

\begin{proof}
Since $0\le \chi_{[I]_{2\theta}} \le 0$, we have $0\le   f_{\xi,\theta}(x)\le \int_\R \Phi_\xi(y)dy=1$. Moreover,  for $x\in I$,
\be
1-  f_{\xi,\theta}(x)&=\int_\R(1-\chi_{[I]_{2\theta}}(y))\Phi_\xi(x-y)dy\le\int_{\R \setminus [I]_{2\theta}}\Phi_\xi(x-y)dy  \\
& \le  2\int_\theta^\infty \Phi_\xi (y)dy 
\le   \tfrac 2{\sqrt{\pi \xi }\theta}\e^{-\xi \theta^2},
\ee
where we used the standard estimate for the  Mills ratio of normal distributions.

Finally, for  $x\in \R \setminus [I]_{3\theta}$, we have
\be
  f_{\xi,\theta}(x)\le\int_{[I]_{2\theta}}\Phi_\xi(x-y)dy\le  2\int_\theta^\infty \Phi_\xi(y)dy   \le    \tfrac 2{\sqrt{\pi \xi }\theta}\e^{-\xi \theta^2}.
\ee
This establishes property (i).

We now turn to property (ii).  Since $\hat   f_{\xi,\theta}=  \sqrt{2\pi} \hat \chi_{[I ]_{2\theta} } \hat \Phi_\xi$
and $\hat \Phi_\xi(s)= \frac {1}{\sqrt{2\pi}} \e^ {-\frac {s^2}{4\xi} }$, we can bound
\be\label{eq:f1}
\norm{\hat   f_{\xi,\theta}}_1\le\sqrt{2\pi}  \norm{\hat \chi_{[I ]_{2\theta}}}_\infty \|\hat \Phi_\xi\|_1\le 
 \norm{\chi_{[I ]_{2\theta}}}_1\|\hat \Phi_\xi\|_1\le   \sqrt{2 \xi}(|I|+4\theta).
\ee
Moreover,
 \be\label{eq:rem}
&\int_{\R \setminus [-\zeta,\zeta]} |\hat   f_{\xi,\theta}(s)|ds \le \sqrt{2\pi}   \|\hat \chi_{[I ]_{2\theta}}\|_\infty \int_{\R \setminus [-\zeta,\zeta]}|\hat \Phi_\xi(s)\d s  \le \tfrac 1 {\sqrt{2\pi}}  \|\chi_{[I ]_{2\theta}}\|_1 
\int_{\R \setminus [-\zeta,\zeta]} \e^{-\frac {s^2}{4\xi}}\d s\\
& \quad  \le 2  (|I|+4\theta) {\sqrt{\tfrac \xi\pi}} \int_{\frac \zeta{\sqrt{2\xi}}}^\infty \e^{-\frac {s^2}{2}}\d s
\le 2 (|I|+4\theta) {\sqrt{\tfrac \xi \pi}} \tfrac {\sqrt{2 \xi}} \zeta \e^{- \frac {\zeta^2}{4\xi}}= 2 \sqrt {\tfrac 2 \pi}  (|I|+4\theta)\tfrac \xi \zeta \e^{- \frac {\zeta^2}{4\xi}}.
\ee
This establishes property (ii).

It remains to prove property (iii).  It follows from \eq{deftildef} that
\be\label{eq:fdiff9}
\norm{\tilde   f_{\xi,\theta,\zeta}-   f_{\xi,\theta}}_\infty\le  \frac 1 {\sqrt{2\pi} }\int_{\R \setminus [-\zeta,\zeta]}  |\hat   f_{\xi,\theta}(s)|ds\le  2 \sqrt {\tfrac 2 \pi}  (|I|+4\theta)\tfrac \xi \zeta \e^{- \frac {\zeta^2}{4\xi}},
\ee 
which is \eq{eq:fdiff5}.
\end{proof}

\begin{lemma}\label{lemfPU}
Fix parameters $\Delta_0>1$ and  $ \lambda_0 >0$.  Let  $q\in \frac 12 \N^0$ and assume that $\Delta \ge \Delta_0$ and $\lambda \ge \lambda_0$ with  $\lambda \Delta^2\ge Y_{q+\frac 12}$.
Let $\xi\in \N $ and set $f_\xi= f_{\xi,\frac 1 6 \nfd} $  and $\tilde f_\xi= f_{\xi,\frac 1 6 \nfd, \frac \Delta 8 \xi} $ as in Lemma \ref{lemfxi} with $I=\tilde I_q$, where $\tilde I_q= [0, (q +\frac 34 )\nfd]$.   Then:

\begin{enumerate}
\item We have
\be\label{hatfxi}
 \|\hat  f_{\xi} \|_1&\le   \sqrt{2\xi}( q + \tfrac 32 )\tfd
 \ee
and
\be\label{eq:fdiff59}
\norm{\tilde   f_{\xi}-   f_{\xi}}_\infty\le   2 \sqrt {\tfrac 2 \pi}  ( q + \tfrac 32 )\tfd\tfrac 8 \Delta  \e^{- \frac {\Delta^2}{256}\xi}\le C_q \e^{- c\xi}.
\ee

\item  Let $K\subset \La \subset \Z$, where $\La$ and $K$ are intervals (we allow $K=\emptyset$ and $\La=\Z$),
and  set    $f_\xi^{\La \upharpoonright K}= f_\xi(H^{\La \upharpoonright K})$ and $\tilde f_\xi^{\La \upharpoonright K}=\tilde  f_\xi(H^{\La \upharpoonright K})$. Let $S$ and $U$ be observables on on $\La$.   Then
\be\label{eq:fdiff77}
&\norm{ (S)_{f_\xi^{\La \upharpoonright K}}-(S)_{\tilde f_\xi^{\La \upharpoonright K}} } \le  C_q \norm{S}\e^{- \frac {\Delta^2}{512}\xi} \le C_q\norm{S}\e^{- c\xi},
\ee 
\be\label{(S)P}
\norm { (S)_{P^{\La \upharpoonright K}_{I_{\le q}}}}\le 2\norm{\pa{S}_{f_\xi^{\La \upharpoonright K}}} \qtx{for} \xi \ge C_1,
\ee
and 
\be\label{YSF}
\norm{U\pa{S}_{f_\xi^{\La}}}
 \le \norm{U\pa{S}_{P^{\La \upharpoonright K}_{I_{\le q+\frac 12}}}} +C \e^{- c \xi}\norm{U} \norm{S}.
\ee
\end{enumerate}
\end{lemma}

\begin{proof}
The estimates \eq{hatfxi} and   \eq{eq:fdiff59}  follow immediately from  \eq{hatfxitheta} and \eq{eq:fdiff5}, respectively, using $\abs{\tilde I_q} =  (q +\frac 34 )\nfd$.

To prove \eq{eq:fdiff77},  note that 
\be
 P_{\le q}(H^{\La \upharpoonright K})= \chi_{\tilde I_q}(H^{\La \upharpoonright K})\qtx{and} P_{\le q+\frac 12 }(H^{\La \upharpoonright K})= 
\chi_{[\tilde I_q]_{\frac 12 \nfd}}(H^{\La \upharpoonright K}),
\ee 
so it  follows from \eq{eq:fdiff59}  and $\norm{f_\xi^{\La \upharpoonright K}}_\infty=1$ that 
\be\label{eq:fdiff77998}
&\norm{ (S)_{f_\xi^{\La \upharpoonright K}}-(S)_{\tilde f_\xi^{\La \upharpoonright K}} } \le 
\norm{S} \norm{\tilde   f_{\xi}-   f_{\xi}}_\infty\pa{\norm{f_\xi}_\infty +  \norm{\tilde f_\xi}_\infty  } \\
&\quad\le  \norm{S}\norm{\tilde   f_{\xi}-   f_{\xi}}_\infty\pa{1 +
\pa{1 +\norm{\tilde   f_{\xi}-   f_{\xi}}_\infty} } \le  C_q \norm{S}\e^{- \frac {\Delta^2}{512}\xi}\le C_q\norm{S}\e^{- c\xi}.
\ee 

Using Lemma~\ref{lemfxi}(i) we immediately get  \eq{(S)P}, and, in addition,
\be\label{YSF23}
\norm{U\pa{S}_{f_\xi^{\La \upharpoonright K}}}&\le  \norm{U\pa{S}_{P^{\La \upharpoonright K}_{I_{\le q+\frac 12}}}} +C \e^{-\frac {\xi\nfd^2} {72}}\norm{U} \norm{S}\\ & 
 \le \norm{U\pa{S}_{P^{\La}_{I_{\le q+\frac 12}}}} +C \e^{- c \xi}\norm{U} \norm{S}.
\ee
\end{proof}

\subsection{Finite speed of propagation}

Let  $\La\subset \Z$ be  an arbitrary interval, possibly $\Z$, and let    $K\subset \La$ be a finite interval (possibly the empty set).
We write  $K^c=\La\setminus K$ when it is clear from the the context.

  \begin{lemma}\label{lemB1}
For every $r\in\N$, $0<\rho < \e^{-1} $, and  $ \abs{s} \le \rho \Delta r$, we have 
\beq
\norm{P_-^A \e^{isH^{\La{\upharpoonright} K}} P_+^{[A]_r}} \le  \tfrac 1{\sqrt{2\pi r}}  \e^{-\ln\pa{\frac 1{\e \rho}}r},
\eeq
for all $A \subset \La \subset \Z$ , where $A$ is a finite interval.

In particular, taking $\rho = \frac 14$ we get
\beq\label{inpart}
\norm{P_-^A \e^{isH^{\La{\upharpoonright} K}} P_+^{[A]_r}} \le  \tfrac 1{\sqrt{2\pi r}}  \e^{-\pa{\ln \frac 4{\e }} r}\le 
 \tfrac 1{\sqrt{2\pi r}}  \e^{-\frac 3 {10} r} \le C  \e^{-\frac 1 {5} r}  \qtx{for} \abs{s} \le \tfrac \Delta 4 r.
\eeq
\end{lemma}

\begin{proof}  
It follows from  \cite[Lemma B.1]{EK22}, which holds for  $  H^{\La{\upharpoonright} K}$ with $\gamma=\frac 1 \Delta$, that
\be
 \norm{P_-^A \e^{isH^{\La{\upharpoonright} K}} P_+^{[A]_r}}\le \tfrac { \abs{s}^r} {\Delta^{r}r!} .
 \ee
If $ \abs{s} \le \rho \Delta r$, we get, using Stirling's formula, that
\be
\norm{P_-^A \e^{isH^{\La{\upharpoonright} K}} P_+^{[A]_r}}\le  \rho^r r^r \pa{\sqrt{2\pi r} \e^{-r}r^r}^{-1}\le \tfrac 1{\sqrt{2\pi r}}  \e^{-\pa{\ln \frac 1{\e \rho}} r}.
\ee  
\end{proof}

\section{ Proof of the main result}
 
Fix parameters $\Delta_0>1$ and  $ \lambda_0 >0$.  Let  $q\in \frac 12 \N^0$, set $p=q + \frac 12$, and assume that $\Delta \ge \Delta_0$ and $\lambda \ge \lambda_0$ with  $\lambda \Delta^2\ge Y_{p}$.
Let $\xi\in \N $ and let $f_\xi$  and $\tilde f_\xi $ be  as in Lemma \ref{lemfPU}. 

Let $T$  be an observable supported on a finite  interval $\mathcal X\subset \Z$ with $\norm{T}\le 1$.  Given $\ell \in \N$ and $t \in \R$, we   take 
 \be\label{defxi}
 \xi =\max \set{  \left\lceil\tfrac {\alpha_{p}} 5\right\rceil \ell, \left\lceil\tfrac 8 \Delta\abs{t}\right\rceil , \cl{C_1}}  , 
 \ee
 where $ {\alpha_{p}} $ is as in Theorem \ref{thm:localmodell} and $C_1$ is the constant in \eq{(S)P}. We observe that
 $[\cX]_{ {\alpha_{p}}\ell} \subset  [\cX]_{ 5\xi}$.   
It follows from Theorem \ref{thm:localmodell} that there exists a random observable ${T}_t={T}(t,p,\ell)$, supported on $ [\cX]_{ 5\xi}$,  satisfying \eq{T_t} and \eq{eq:locality2mgq} with $p$ substituted for $q$.
 
   We will prove that  ${T}_t={T}(t, p,\ell)$ satisfies \eqref{eq:main2}, proving the theorem. The proof  relies on the three guidance principles introduced in the previous section, and proceeds in several steps, as follows:

\begin{enumerate}
\item  
We  start by  establishing the bound
\be\label{SPP+}
\E\norm{\pa{S}_{P^\Z_{I_{\le q}}}}  \le  2\sum_{j=1}^{ \cl{p}+5}  \E\norm{P_+^{\partial_{3{\xi}}[\mathcal X]_{9j{\xi}}} \pa{S}_{f_\xi^\Z}} 
+    C_{q,\abs{\cX}}  \norm{S}  \e^{-c_q \xi},
\ee
for all  observables $S$ on  $\Z$ and finite intervals   $\mathcal X\subset \Z$.

\item We use the estimates from  Theorem \ref{thm:localmodell}, in the form
\be\label{normTt}
 \norm{T_t}\le   C_{p} \langle t\rangle^{\gamma_{p}},
  \ee 
 and
 \be\label{eq:locres}
 &\E\norm{ \pa{\tau^{[\mathcal X]_{9j\xi}}_t(T)-{T}_t}_{P_{I_{\le  p}}^{[\mathcal X]_{9j\xi}}} }
\le  C_{p}\langle t\rangle^{\gamma_{p}}\pa{\abs{\cX}+ 18j \xi}^{\rho_{p} }\e^{- \theta_{p}\, \ell}\\
&   \hskip50pt   \le C_{q,\abs{\cX}}\langle t\rangle^{\gamma^\pr_{q}}e^{- \frac 1 2 \theta_{p}\ell}
\qtx{for all} j=1,2,\dots,  \cl{p} +5,
\ee
where $\gamma^\pr_{q}= \gamma_{p}+\rho_{p} \ge \gamma_{p}$ ,
to derive the estimate
\be\label{P+taut-TK}
\E\norm{ \pa{\tau^{\upharpoonright {[\mathcal X]_{9j\xi}}}_t(T)-{T}_t}_{f_\xi^{\upharpoonright {[\mathcal X]_{9j\xi}}}}}
\le C_{q,\abs{\cX}}\langle t\rangle^{\gamma^\pr_{q}}e^{- \theta^{\pr}_{q} \ell}\mqtx{for all} j=1,2,\dots,  \cl{p} +5.
\ee

\item We use  \eq{P+taut-TK}
to  prove the estimate
\be\label{P+tauTTt5}
&\E\norm{ P_+^{\partial_{3{\xi}}[\mathcal X]_{9j{\xi}}}\pa{\tau_t(T)-{T}_t}_{f^\Z_\xi}}\\ & 
\qquad  \le \E\norm{ P_+^{\partial_{3{\xi}}[\mathcal X]_{9j{\xi}}}\pa{\pa{\tau_t(T)-{T}_t}_{\tilde f^\Z_\xi}-  \pa{\tau^{\upharpoonright {[\mathcal X]_{9j\xi}}}_t(T)-{T}_t}_{\tilde f_\xi^{\upharpoonright {[\mathcal X]_{9j\xi}}}}  } P_+^{\partial_{{\xi}}[\mathcal X]_{9j{\xi}}}} \\
&  \hskip90pt  +C_{q,\abs{\cX }} \langle t\rangle^{\gamma^\pr_{q}}e^{-\theta^{\prr\pr}_{q}\ell} \qtx{for all} j=1,2,\dots,  \cl{p} +5.
\ee

\item  We  show that
\be\label{Iest2}
&\E\norm{ P_+^{\partial_{3{\xi}}[\mathcal X]_{9j{\xi}}}\pa{\pa{\tau_t(T)}_{\tilde f_\xi}-  \pa{\tau^{\upharpoonright {[\mathcal X]_{9j\xi}}}_t(T)}_{\tilde f_\xi^{\upharpoonright {[\mathcal X]_{9j\xi}}}}  } P_+^{\partial_{{\xi}}[\mathcal X]_{9j{\xi}}}}\\
& \qquad +\E\norm{ P_+^{\partial_{3{\xi}}[\mathcal X]_{9j{\xi}}}\pa{\pa{{T}_t}_{\tilde f_\xi}-  \pa{{T}_t}_{\tilde f_\xi^{\upharpoonright {[\mathcal X]_{9j\xi}}}}  } P_+^{\partial_{{\xi}}[\mathcal X]_{9j{\xi}}}} \le C_q  \langle t\rangle^{\gamma_{p}+1}e^{-c\xi}.
\ee

\item

Finally, combining \eqref{defxi}, \eq{SPP+}, \eq{normTt}, \eqref{P+tauTTt5},   and \eqref{Iest2}, we get \eqref{eq:main2}.
\end{enumerate}

We start by  proving \eq{SPP+}. 
Let $S$ be an observable on $\Z$ and  $\mathcal X\subset \Z$   a finite  interval.  It follows from \eq{(S)P}, 
 the identity
\be
\mathds{1}_{\cH_\Z}=\prod_{i=1}^{  k+5}P_-^{\partial_{3{\xi}}[\mathcal X]_{9i{\xi}}}\,+\,\sum_{j=1}^{  k+5}P_+^{\partial_{3{\xi}}[\mathcal X]_{9j{\xi}}} \prod_{i=1}^{j-1}P_-^{\partial_{3{\xi}}[\mathcal X]_{9i{\xi}}},
\ee  
where $  {\prod_{i=1}^{0}P_-^{\partial_{3{\xi}}[\mathcal X]_{9i{\xi}}} = \ident_{\cH_\Z}}$,  and \eq{YSF} that
\be
&\E\norm{\pa{S}_{P^\Z_{I_{\le q}}}} \le 2 \E\norm{\pa{S}_{f_\xi^\Z}} \le 
\\&      \le 2 \E\norm{\sum_{j=1}^{ \cl{p}+5}\prod_{i=1}^{  j-1}P_-^{\partial_{3{\xi}}[\mathcal X]_{9i{\xi}}}P_+^{\partial_{3{\xi}}[\mathcal X]_{9j{\xi}}} \pa{S}_{f_\xi^\Z}}+ 
2\E\norm{\prod_{i=1}^{ \cl{p} +5}P_-^{\partial_{3{\xi}}[\mathcal X]_{9i{\xi}}}\pa{S}_{f_\xi^\Z}}
\\ & \le 2\sum_{j=1}^{ \cl{p}+5}  \E\norm{P_+^{\partial_{3{\xi}}[\mathcal X]_{9j{\xi}}} \pa{S}_{f_\xi^\Z}} + 2\E\norm{\prod_{i=1}^{ \cl{p} +5}P_-^{\partial_{3{\xi}}[\mathcal X]_{9i{\xi}}}\pa{S}_{P^\Z_{I_{\le  p}}}}   +C\norm{S} \e^{-c \xi }.
\ee

We now apply  Lemma \ref{modLemma} with $k=\cl{p}$, $\La=  [\mathcal X]_{9(q+ 8){\xi}}$,  $S_i={\partial_{3{\xi}}[\mathcal X]_{9i{\xi}}}$ for $i=1,2,\ldots, \cl{p} +5$, $W_{max}= 2$, and   $L= \fl{\frac{3\xi-1}2}\ge \xi$,   getting
\be
&\E\norm{\prod_{i=1}^{ \cl{p} +5}P_-^{\partial_{3{\xi}}[\mathcal X]_{9i{\xi}}}\pa{S}_{P^\Z_{I_{\le  p}}}}\le  2C_{\cl{q+ \frac 12}}\norm{S}\pa{\abs{\cX}+ 18(q+ 8){\xi}}^{2(\cl{p}) +15} }\e^{- c_{\cl{p}}{\xi} \\
& \qquad  \qquad \le  \norm{S} C^\pr_{q,\abs{\cX}}{\xi}^{2(\cl{p}) +15} \e^{- c_{\cl{p}}{\xi} }
\le  \norm{S} C^{\prr}_{q,\abs{\cX}}\e^{-\frac 12 c_{\cl{p}}\xi}.
\ee
 The equation \eq{SPP+} follows.

To prove \eq{P+taut-TK}, we 
fix $j\in\set{1,2,\dots,  \cl{p} +5}$.  Using  \eq{YSF}, \eq{PLaKK}, \eq{normTt}, \eq{suppsK}, \eq{eq:locres}, we get
\be\label{P+taut-T222}
& \E\norm{{f_\xi^{\upharpoonright {[\mathcal X]_{9j\xi}}}} \pa{\tau^{\upharpoonright {[\mathcal X]_{9j\xi}}}_t(T)-{T}_t}{f_\xi^{\upharpoonright {[\mathcal X]_{9j\xi}}}}}\\
& \qquad  \le \E\norm{ {P_{I_{\le  p}}^{\upharpoonright {[\mathcal X]_{9j\xi}}}}\pa{\tau^{\upharpoonright {[\mathcal X]_{9j\xi}}}_t(T)-{T}_t}{P_{I_{\le  p}}^{\upharpoonright {[\mathcal X]_{9j\xi}}}}} +C(1+\norm{T_t}) \e^{-c \xi}\\
& \qquad = \E\norm{ {P_{I_{\le p }}^{[\mathcal X]_{9j\xi}}P_{I_{\le p }}^{\upharpoonright {[\mathcal X]_{9j\xi}}}}\pa{\tau^{\upharpoonright  [\mathcal X]_{9j\xi}}_t(T)-{T}_t}{P_{I_{\le p }}^{[\mathcal X]_{9j\xi}}P_{I_{\le p }}^{\upharpoonright {[\mathcal X]_{9j\xi}}}}} + C_{q}  \langle t\rangle^{\gamma_{p}} \e^{- c \xi}\\
& \qquad \le \E\norm{ {P_{I_{\le p }}^{[\mathcal X]_{9j\xi}}}\pa{\tau^{[\mathcal X]_{9j\xi}}_t(T)-{T}_t}{P_{I_{\le p }}^{[\mathcal X]_{9j\xi}}}} +C_{q}  \langle t\rangle^{\gamma_{p}} \e^{- c \xi}\\
& \qquad \le   C_{q,\abs{\cX}}\langle t\rangle^{\gamma^\pr_{q}}e^{- \frac 1 2 \theta_{p} \ell} +C_{q}  \langle t\rangle^{\gamma_{p}} \e^{- c \xi}
\le C_{q,\abs{X}}^\pr \langle t\rangle^{\gamma_{q}^\pr}e^{-\theta^\pr_{q}\, \ell},
\ee
which is \eq{P+taut-TK}.

We now turn to the proof of \eq{P+tauTTt5}.    Let $Y$ be an observable on $\Z$ with $\norm{Y}\le 1$.   Using first \eq{P+taut-TK}, followed by \eqref{eq:fdiff77},\eqref{normTt}, and \eq{defxi},  we get
\be\label{Pfxitau}
&\E\norm{Y\pa{\tau_t(T)-{T}_t}_{f_\xi}}\\
& \;  \le 
\E\norm{Y\pa{\pa{\tau_t(T)-{T}_t}_{f_\xi}-  \pa{\tau^{\upharpoonright {[\mathcal X]_{9j\xi}}}_t(T)-{T}_t}_{f_\xi^{\upharpoonright {[\mathcal X]_{9j\xi}}}}  } }  +  C_{q,\abs{X}}^\pr \langle t\rangle^{\gamma_{q}^\pr}\e^{-\theta^\pr_{q}\, \ell} \\
& \;  \le 
\E\norm{Y\pa{\pa{\tau_t(T)-{T}_t}_{\tilde f_\xi}-  \pa{\tau^{\upharpoonright {[\mathcal X]_{9j\xi}}}_t(T)-{T}_t}_{\tilde f_\xi^{\upharpoonright {[\mathcal X]_{9j\xi}}}}  } }  +  C_q  \langle t\rangle^{\gamma_{p}}  \e^{-c\xi} +  C_{q,\abs{X}}^\pr \langle t\rangle^{\gamma_{q}^\pr}\e^{-\theta^\pr_{q}\, \ell}\\
& \; \le 
\E\norm{Y\pa{\pa{\tau_t(T)-{T}_t}_{\tilde f_\xi}-  \pa{\tau^{\upharpoonright {[\mathcal X]_{9j\xi}}}_t(T)-{T}_t}_{\tilde f_\xi^{\upharpoonright {[\mathcal X]_{9j\xi}}}}  } }  +  C_{q,\abs{X}} \langle t\rangle^{\gamma_{q}^\pr} \e^{-\theta^{\prr}_{q}\, \ell}.
\ee

 We now let  $\#$ be either $\Z$  or ${\upharpoonright \!\![\mathcal X]_{9j\xi}}$, and 
 consider $\tilde f_\xi^\#$, $H^\#$, $\tau^\#_t(\cdot)$.   Let $S$ be either $T$  or $T_t$.   Since $S$ is supported by the interval  $[\cX]_{ 5\xi}$,  we have
\be\label{SPpm}
[S, P_\pm^{\partial_{2\xi}[\mathcal X]_{9j{\xi}}}]=0.
\ee
Let $s\in [-t,t]$, so  \eq{defxi} yields $\abs{s}\le \abs{t} \le \frac \Delta 8 \xi$.
Then it follows from  \eq{inpart} in Lemma \ref{lemB1}, using $\ident=P_+^{\partial_{2{\xi}}[\mathcal X]_{9j{\xi}}}+P_-^{\partial_{2{\xi}}[\mathcal X]_{9j{\xi}}}$ and \eq{SPpm},  that
\be
&\E\norm{ P_+^{\partial_{{3\xi}}[\mathcal X]_{9j{\xi}}}
\pa{\tau_s^\#(S)}_{\tilde f_\xi^\#}P_-^{\partial_{{\xi}}[\mathcal X]_{9j{\xi}}} }\\
& \quad \le \norm{S} \pa{ \E\norm{ P_+^{\partial_{{3\xi}}[\mathcal X]_{9j{\xi}}}\tilde f_\xi^\#\e^{isH^\#}P_-^{\partial_{2{\xi}}[\mathcal X]_{9j{\xi}}}} +\E\norm{ P_+^{\partial_{{2\xi}}[\mathcal X]_{9j{\xi}}} \e^{-isH^\#}{\tilde f_\xi^\#}P_-^{\partial_{{\xi}}[\mathcal X]_{9j{\xi}}} }}\\
&\quad \le \tfrac 1 {\sqrt{2\pi}}\norm{S}  \int_{[-\frac \Delta 8 \xi,\frac \Delta 8 \xi]}\left(\E\norm{ P_+^{\partial_{{3\xi}}[\mathcal X]_{9j{\xi}}}\e^{i(s-u)H^\#}P_-^{\partial_{2{\xi}}[\mathcal X]_{9j{\xi}}}} \right.  \\
& \hskip140pt
 \left.  + \E\norm{ P_+^{\partial_{{2\xi}}[\mathcal X]_{9j{\xi}}} \e^{-i(u+s)H^\#}P_-^{\partial_{{\xi}}[\mathcal X]_{9j{\xi}}} }\right)  \abs{\hat   f_{\xi}(u)} \, \d u\\
 & \quad   \le   C \norm{S} \norm{\hat   f_{\xi}}_1 \e^{-\frac 2 {5} \xi} \le   C \norm{S}  \sqrt{ 2\xi} (q +\tfrac 32 )\tfd\e^{-\frac 2 {5} \xi}
 \le  C^\pr_q \norm{S}\e^{-\frac 1 {5} \xi} ,
\ee
where we used also \eq{hatfxi}.
This, together with \eqref{normTt}, yields the estimate
\be\label{EPtauTt}
\E\norm{ P_+^{\partial_{{3\xi}}[\mathcal X]_{9j{\xi}}}
\pa{\tau^\#_t(T)-{T}_t}_{\tilde f_\xi^\#}P_-^{\partial_{{\xi}}[\mathcal X]_{9j{\xi}}} } \le  C_q  \langle t\rangle^{\gamma_{p}}
  \e^{-\frac 1 {5} \xi}.
\ee

It  now  follows from  \eq{Pfxitau}, using the  identity $\ident=P_+^{\partial_{{\xi}}[\mathcal X]_{9j{\xi}}}+P_-^{\partial_{{\xi}}[\mathcal X]_{9j{\xi}}}$, the triangular inequality, and \eq{EPtauTt}, that
\be\label{Pfxitau24}
&\E\norm{P_+^{\partial_{{3\xi}}[\mathcal X]_{9j{\xi}}}\pa{\tau_t(T)-{T}_t}_{f_\xi}} \\
& \quad  \le 
\E\norm{P_+^{\partial_{{3\xi}}[\mathcal X]_{9j{\xi}}}\pa{\pa{\tau_t(T)-{T}_t}_{\tilde f_\xi}-  \pa{\tau^{\upharpoonright {[\mathcal X]_{9j\xi}}}_t(T)-{T}_t}_{\tilde f_\xi^{\upharpoonright {[\mathcal X]_{9j\xi}}}}  } P_+^{\partial_{{\xi}}[\mathcal X]_{9j{\xi}}}} \\
&\hskip80pt + +  C^\pr_{q,\abs{X}} \langle t\rangle^{\gamma_{q}^\pr} \e^{-\theta^{\prr\pr}_{q}\, \ell},
\ee
which is \eq{P+tauTTt5}.

To establish \eqref{Iest2},
let $s\in [-t,t]$ and  let $S$ be an observable. Let $s\in [-t,t]$ and  let $S$ be an observable. Then, 
\be\label{def(III)}
& (I)= \E\norm{ P_+^{\partial_{3{\xi}}[\mathcal X]_{9j{\xi}}}\pa{\pa{\tau_s(S)}_{\tilde f_\xi}-  \pa{\tau^{{\upharpoonright {[\mathcal X]_{9j\xi}}}}_{s}(S)}_{\tilde f_\xi^{{\upharpoonright {[\mathcal X]_{9j\xi}}}} }}P_+^{\partial_{{\xi}}[\mathcal X]_{9j{\xi}}}}  \\
&\le  \tfrac 1 {{2\pi}}    \int_{-\frac \Delta 8 \xi}^{\frac \Delta 8 \xi} \int_{-\frac \Delta 8 \xi}^{\frac \Delta 8 \xi}
\E\norm{ P_+^{\partial_{3{\xi}}[\mathcal X]_{9j{\xi}}}\Psi(S,s-u,s+v)P_+^{\partial_{{\xi}}[\mathcal X]_{9j{\xi}}}}   \abs{\hat   f_{\xi}(u)} \abs{\hat   f_{\xi}(v)}\, \d u \d v,
\ee
where
\be
\Psi(S,u,v)=\e^{iuH} S \e^{-ivH} -  \e^{iuH^{{\upharpoonright {[\mathcal X]_{9j\xi}}}}} S \e^{-ivH^{{\upharpoonright {[\mathcal X]_{9j\xi}}}}}  \qtx{for} u, v \in \R.
\ee

By Duhamel's Formula, observing that 
$\Gamma^{{\upharpoonright {[\mathcal X]_{9j\xi}}}}=P_-^{\partial[\mathcal X]_{9j{\xi}}}\Gamma^{{\upharpoonright {[\mathcal X]_{9j\xi}}}} P_- ^{\partial [\mathcal X]_{9j{\xi}}}$, we have
\be
 \e^{irH}-\e^{irH^{{\upharpoonright {[\mathcal X]_{9j\xi}}}}}
&  =i \int_0^r\e^{i(r-{ x})H^{{\upharpoonright {[\mathcal X]_{9j\xi}}}}}\Gamma^{{\upharpoonright {[\mathcal X]_{9j\xi}}}} \e^{i{ x}H} \, \d { x}\\&=i \int_0^r\e^{i(r-{ x})H^{{\upharpoonright {[\mathcal X]_{9j\xi}}}}}P_-^{\partial[\mathcal X]_{9j{\xi}}}\Gamma^{{\upharpoonright {[\mathcal X]_{9j\xi}}}} P_- ^{\partial [\mathcal X]_{9j{\xi}}} \e^{i{ x}H} \, \d { x}.
\ee
 Letting
\be
\Upsilon(r,x)=\e^{i(r)H^{{\upharpoonright {[\mathcal X]_{9j\xi}}}}}P_-^{\partial[\mathcal X]_{9j{\xi}}}\Gamma^{{\upharpoonright {[\mathcal X]_{9j\xi}}}} P_- ^{\partial [\mathcal X]_{9j{\xi}}} \e^{i{x}H} \qtx{for} r,x \in \R,
\ee
and  recalling \eq{GammaK}, we get
 \be\label{(IV)}
(II)& =  
\norm{ P_+^{\partial_{3{\xi}}[\mathcal X]_{9j{\xi}}}{\Psi(S,s-u,s+v)}P_+^{\partial_{{\xi}}[\mathcal X]_{9j{\xi}}}}\\
& \le   \abs{ \int_0^{s-u} \norm{ P_+^{\partial_{3{\xi}}[\mathcal X]_{9j{\xi}}}
\pa{\Upsilon(s-u-x,x) S \e^{-i(s+v)H}}P_+^{\partial_{{\xi}}[\mathcal X]_{9j{\xi}}}}\d {x}     } \\
& \qquad + \abs{\int_0^{{-(s+v)}} \norm{ P_+^{\partial_{3{\xi}}[\mathcal X]_{9j{\xi}}}\pa{\e^{i(s-u)H^{{\upharpoonright {[\mathcal X]_{9j\xi}}}}}S 
\Upsilon(-s-v-x,-x) }P_+^{\partial_{{\xi}}[\mathcal X]_{9j{\xi}}}}\d {x} } \\
& \le 2 \norm{S}  \abs{ \int_0^{s-u} \norm{ P_+^{\partial_{3{\xi}}[\mathcal X]_{9j{\xi}}}{\e^{i(s-u-{x})H^{{\upharpoonright {[\mathcal X]_{9j\xi}}}}}P_-^{\partial[\mathcal X]_{9j{\xi}}}    }}\d {x} } \\ &\qquad +2 \norm{S}
\int_0^{-(s+v)} \abs{\norm{  P_- ^{\partial [\mathcal X]_{9j{\xi}}} \e^{-i{x}H} P_+^{\partial_{{\xi}}[\mathcal X]_{9j{\xi}}}}\d {x}}.
\ee

Using \eq{inpart} in Lemma \ref{lemB1}  we get
\be\label{E(IV)}
\E  (II)\le  2C \norm{S} \abs{t} \pa{\e^{-\frac 3 5 \xi}+\e^{-\frac 1 5 \xi}}\le 
C^\pr \norm{S}\abs{t}\e^{-\frac 1 5 \xi}.
\ee

Combining \eq{def(III)}, \eq{(IV)}, \eq{E(IV)}, and   \eq{hatfxitheta},  we get
\be\label{IIIest}
 (I)\le  \tfrac C {{2\pi}} \tfrac{ \Delta^2} {16} \pa{\sqrt{ 2\xi} ((q +\tfrac 34 )\tfd+\tfrac 23 \tfd)}^2\norm{S}\abs{t}\e^{-\frac 1 5 \xi}\le C_q \norm{S} \abs{t}\e^{-\frac 1 {10}\xi},
\ee
 for all $s\in [-t,t]$. 
Putting together  \eq{def(III)}, \eq{IIIest},  and \eq{normTt} we get \eqref{Iest2}.

Finally, combining \eqref{defxi}, \eq{SPP+}, \eqref{P+tauTTt5},   and \eqref{Iest2},   we  get \eqref{eq:main2}. 
Theorem \ref{thm:main2} is proven. 
\section*{Declarations}

\subsection*{\qquad Data availability}
We do not analyze or generate any datasets, because our work proceeds within a theoretical and mathematical approach. One can obtain the relevant materials from the references below.

\subsection *{\qquad  Funding and/or Conflicts of interests/Competing interests} \

Alexander Elgart was  supported in part by the NSF under grant DMS-2307093.

The authors have no relevant financial or non-financial interests to disclose.

The authors have no competing interests to declare that are relevant to the content of this article.
\printbibliography

\end{document}